\pdfoutput=1
\documentclass[11pt]{amsart}

\usepackage{amsmath}
\usepackage{amsthm}
\usepackage{graphicx}
\usepackage{booktabs}
\usepackage{array}
\usepackage{url}
\usepackage{amssymb}

\newtheorem{definition}{Definition}[section]
\newtheorem{proposition}{Proposition}[section]
\newtheorem{theorem}{Theorem}[section]
\newtheorem{lemma}{Lemma}[section]
\newtheorem{assumption}{Assumption}[section]
\newtheorem{remark}{Remark}[section]

\newcommand{\E}{\mathbb{E}}
\newcommand{\R}{\mathbb{R}}
\newcommand{\ii}{\mathrm{i}}

\newcommand{\Cov}{\operatorname{Cov}}
\newcommand{\Rea}{\operatorname{Re}}
\newcommand{\Ima}{\operatorname{Im}}
\newcommand{\argmin}{\operatorname*{arg\,min}}

\begin{document}

\title[Moment-free Kunchenko stochastic polynomials]{Moment-Free Kunchenko
Stochastic Polynomials via Empirical Characteristic Functions}

\author{S. Zabolotnii}
\address{Cherkasy State Business College, Cherkasy 18028, Ukraine}
\address{State Scientific Research Institute of Armament and Military Equipment
Testing and Certification, Cherkasy, Ukraine}
\address{Uzhhorod National University, Uzhhorod, Ukraine}
\curraddr{}
\email{zabolotnii.serhii@csbc.edu.ua}

\date{}
\dedicatory{}
\thanks{The Lean 4 material cited in the paper is used only as an independent
check of deterministic algebraic identities.  All probabilistic arguments are
conventional.}

\begin{abstract}
We give a characteristic-function formulation of Kunchenko's
stochastic-polynomial construction for settings in which raw moments may fail
to exist.  In the finite-variance trigonometric case, the coefficients of the
Kunchenko normal system are expressed through the characteristic
function and its derivative.  In the moment-free case, empirical characteristic
functions on a fixed finite frequency grid define a bounded discrepancy
geometry that remains meaningful for Cauchy, symmetric stable, and other
heavy-tailed laws.  We prove well-definedness and finite-grid almost sure
consistency of this empirical characteristic-function geometry.  We introduce
the associated minimum-CF-distance estimator and establish its identifiability,
strong consistency, and asymptotic normality on a fixed grid, with a covariance
built from bounded trigonometric moments that stays finite even for Cauchy and
stable laws; refining the grid increases the optimal-weight information
monotonically to the Fisher information, so the estimator is asymptotically
efficient in the dense-grid limit.  We also relate bounded sine scores to weak
stochastic-polynomial estimating equations.  A
small Lean 4 / Mathlib supplement checks selected deterministic identities
underlying the bounded-score construction; convergence arguments and
statistical interpretation remain outside the formalization.
\end{abstract}

\maketitle

\section{Introduction}

The stochastic-polynomial apparatus of Yu.\,P.~Kunchenko rests on a space with a
generating element, on moment and cumulant (correlation) characteristics, and on
normal equations for polynomial-type
approximations~\cite{kunchenko2003,kunchenko2006}.  Built on this foundation, it
has grown along three branches: parameter estimation by the polynomial
maximization method~\cite{kunchenko2002,zabolotnii2019}; statistical detection and change-point
analysis by polynomial decision rules; and pattern recognition by decomposition
in the space with a generating element.  Each branch is effective when the
moments it requires exist and are numerically stable, and each is less
satisfactory for distributions with heavy tails.  For Cauchy laws, stable laws
with characteristic exponent below two, and many finite samples drawn from
moment-unstable models, variance- and cumulant-based descriptions either have no
mathematical meaning or become dominated by rare observations; the heavy-tailed
and stable models at issue are surveyed by Samorodnitsky and
Taqqu~\cite{samorodnitsky1994} and by Resnick~\cite{resnick2007}.

Characteristic functions provide a natural way to separate two issues that are
often conflated.  First, in the finite-moment trigonometric regime, the
Kunchenko normal equations can already be written in the language of the
characteristic function~\cite{zabolotnii2010trig}.  Second, in the moment-free regime, the empirical
characteristic function is a bounded object on every finite frequency grid, so
it can define a finite-dimensional discrepancy even when no variance exists.
The two constructions share characteristic-function evaluations, but they do
not optimize the same functional.

The characteristic-function discrepancy used below is related to several
established constructions.  Estimation by minimizing an integrated squared
distance between empirical and model characteristic functions goes back to
Heathcote~\cite{heathcote1977} and, for stable laws, to
Press~\cite{press1972} and Koutrouvelis~\cite{koutrouvelis1980}; its efficiency
and its generalized-method-of-moments interpretation are due to Feuerverger and
McDunnough~\cite{feuerverger1981} and to Carrasco and
Florens~\cite{carrasco2000}.  These ideas extend to dynamic models through a
continuum of moment conditions~\cite{carrasco2007} and, when the characteristic
function has no closed form, to simulation-based and indirect-inference
estimators~\cite{gourieroux1993}; the heavy-tailed families treated here have
closed-form characteristic functions, so the estimator below evaluates them
directly on a finite grid and needs no simulation.  Weighted squared CF
differences also underlie goodness-of-fit and normality tests, such as the
Epps--Pulley statistic~\cite{epps1983} and the Baringhaus--Henze
test~\cite{baringhaus1988}, and, for a translation-invariant characteristic
kernel, they coincide with the maximum mean discrepancy of kernel two-sample
testing~\cite{sriperumbudur2010, gretton2012} and with the energy distance, the
distance-based and kernel forms being equivalent~\cite{szekely2013,
sejdinovic2013}.  Our aim is not to extend this body of methods but to place the
Kunchenko stochastic-polynomial construction within it: we show that its
trigonometric normal system and a moment-free finite-grid CF discrepancy are
two faces of the same characteristic-function description, and we give the
moment-free asymptotics of the associated estimator.  The asymptotic theory we
invoke is the standard generalized-method-of-moments and
empirical-characteristic-function machinery; what is new is the bridge to
Kunchenko stochastic polynomials, the bounded moment-free finite-grid pathway
valid for laws without variance, and a machine-checked deterministic core for
its algebraic identities.

The contribution of this paper is deliberately narrow.  We formulate a
two-regime characteristic-function extension of Kunchenko stochastic-polynomial
constructions; express the trigonometric finite-variance normal-system terms
through the characteristic function and its derivative; define a finite-grid
empirical characteristic-function distance that requires no moment assumptions
and prove its boundedness and finite-grid consistency; establish
identifiability, strong consistency, and asymptotic normality of the associated
minimum-CF-distance estimator on a fixed grid, with a covariance built from
bounded trigonometric moments and hence valid for Cauchy and stable laws; prove,
by an elementary projection argument, that refining the grid raises the optimally
weighted information monotonically to the Fisher information, so that the
estimator is asymptotically efficient in the dense-grid limit; and isolate
bounded sine scores as a weak stochastic-polynomial estimating equation.  The
paper does not claim a complete characteristic-function theory of parameter
estimation, pattern classification, or change-point detection, nor does it use
minimum mean-square-error language in infinite-variance settings.

The proof style is conventional.  Lean 4 is used only as a reproducibility audit
for deterministic algebraic identities such as boundedness and oddness of the
sine score, a Cauchy-score identity, and a linear-algebra identifiability core.
The law of large numbers, differentiability of characteristic functions, and
statistical consistency statements are proved in the usual mathematical form.

\section{Generating-element systems}

Let $X$ be a real random variable and let
$\varphi_1,\ldots,\varphi_S$ be real basis functions.  In Kunchenko's
generating-element construction~\cite{kunchenko2003,kunchenko2006} one
approximates $X$ by a centered linear
combination of these basis functions,
\[
  X \approx K_0+\sum_{j=1}^S K_j\varphi_j(X),
\]
where the coefficients are obtained from a covariance normal system when the
required second moments are finite.

\begin{definition}[Finite-variance stochastic-polynomial system]\label{def:fv}
Assume that $\E X^2<\infty$ and
$\E\varphi_j(X)^2<\infty$, $j=1,\ldots,S$.  Put
\[
  F_{jk}=\Cov(\varphi_j(X),\varphi_k(X)),\qquad
  B_j=\Cov(X,\varphi_j(X)).
\]
If $F$ is nonsingular, the Kunchenko coefficient vector
$K=(K_1,\ldots,K_S)^T$ is defined by
\[
  FK=B.
\]
The intercept is
\[
  K_0=\E X-\sum_{j=1}^S K_j\E\varphi_j(X).
\]
\end{definition}

The associated residual mean-square error and the usual information functional
are meaningful only under finite second moments.  This restriction is not a
technical nuisance; it is the reason for introducing the moment-free
characteristic-function regime below.

\section{Trigonometric systems through characteristic functions}

Let
\[
  f_X(u)=\E\exp(\ii uX),\qquad u\in\R,
\]
be the characteristic function of $X$~\cite{lukacs1970}.  If $\E|X|<\infty$, then
\[
  f'_X(u)=\ii\,\E\{X\exp(\ii uX)\}.
\]
Under the finite-variance assumptions of Definition~\ref{def:fv} this identity applies
to the trigonometric normal system.

\begin{proposition}[Trigonometric coefficients]\label{prop:trig}
Let $\E X^2<\infty$, let $p>0$, and let $r,k$ be positive integers.  For the
cosine basis $\varphi_r(X)=\cos(rpX)$,
\[
  \E\cos(rpX)=\Rea f_X(rp),
  \qquad
  \E\{X\cos(rpX)\}=\Ima f'_X(rp),
\]
and hence
\[
  B_r^{(c)}
    =\Ima f'_X(rp)-\E X\,\Rea f_X(rp).
\]
Moreover,
\[
  F_{rk}^{(c)}
    =\frac{1}{2}\Rea\{f_X((r-k)p)+f_X((r+k)p)\}
     -\Rea f_X(rp)\Rea f_X(kp).
\]
For the sine basis $\varphi_r(X)=\sin(rpX)$,
\[
  \E\sin(rpX)=\Ima f_X(rp),
  \qquad
  \E\{X\sin(rpX)\}=-\Rea f'_X(rp),
\]
and therefore
\[
  B_r^{(s)}
    =-\Rea f'_X(rp)-\E X\,\Ima f_X(rp),
\]
with
\[
  F_{rk}^{(s)}
    =\frac{1}{2}\Rea\{f_X((r-k)p)-f_X((r+k)p)\}
     -\Ima f_X(rp)\Ima f_X(kp).
\]
\end{proposition}

\begin{proof}
The identities for the expectations of sine and cosine follow by taking the
real and imaginary parts of $f_X(rp)$.  Since
$f'_X(u)=\ii \E\{X\exp(\ii uX)\}$, we have
\[
  \E\{X\exp(\ii uX)\}=\frac{1}{\ii}f'_X(u)=-\ii f'_X(u).
\]
Taking real and imaginary parts yields
$\E\{X\cos(uX)\}=\Ima f'_X(u)$ and
$\E\{X\sin(uX)\}=-\Rea f'_X(u)$.  The formulas for $F_{rk}$ are obtained from
\[
  \cos a\cos b=\frac{1}{2}\{\cos(a-b)+\cos(a+b)\},
\]
and
\[
  \sin a\sin b=\frac{1}{2}\{\cos(a-b)-\cos(a+b)\},
\]
followed by subtraction of the products of the corresponding means.
\end{proof}

\begin{remark}
The derivative $f'_X$ is used only in the finite-first-moment, and here
finite-variance, regime.  The moment-free construction in the next section
does not require differentiability of $f_X$ and does not use covariance
normal equations.
\end{remark}

\section{Finite-grid empirical characteristic-function geometry}

Let $X_1,\ldots,X_n$ be observations from an arbitrary real distribution.  The
empirical characteristic function~\cite{feuerverger1977} is
\[
  \widehat f_n(u)=\frac{1}{n}\sum_{j=1}^n \exp(\ii uX_j).
\]
Fix a finite frequency grid $U=\{u_1,\ldots,u_M\}$ and positive weights
$w_1,\ldots,w_M$.  The weights are often normalized so that
$\sum_m w_m=1$, but normalization is not needed for the definition.

\begin{definition}[Finite-grid CF distance]\label{def:cfdist}
Let $g$ be a reference characteristic function.  Define
\[
  D_{U,w}(\widehat f_n,g)
    =\sum_{m=1}^M w_m|\widehat f_n(u_m)-g(u_m)|^2.
\]
For $\varepsilon>0$, define the log contrast
\[
  L_{U,w,\varepsilon}(\widehat f_n,g)
    =\log\{D_{U,w}(\widehat f_n,g)+\varepsilon\}.
\]
\end{definition}

This is not an $L^2$ approximation error in the original random-variable
space.  It is a finite-dimensional discrepancy between characteristic-function
vectors.  That distinction is essential in infinite-variance settings.

\begin{proposition}[Moment-free well-definedness]\label{prop:bdd}
For every sample $X_1,\ldots,X_n$, every finite grid $U$, every positive
weight vector $w$, and every characteristic function $g$,
\[
  0\le D_{U,w}(\widehat f_n,g)\le 4\sum_{m=1}^M w_m.
\]
If $\sum_m w_m=1$, then $D_{U,w}(\widehat f_n,g)\le 4$.
\end{proposition}

\begin{proof}
For every $u$ and every observation $X_j$,
$|\exp(\ii uX_j)|=1$.  Hence $|\widehat f_n(u)|\le 1$.  Also every
characteristic function satisfies $|g(u)|\le 1$.  Therefore
$|\widehat f_n(u_m)-g(u_m)|\le 2$, and the stated bound follows after
squaring and summing with the weights.
\end{proof}

\begin{proposition}[Finite-grid consistency]\label{prop:consist}
Let $X_1,X_2,\ldots$ be i.i.d. with characteristic function $f_X$.
For every fixed finite grid $U=\{u_1,\ldots,u_M\}$,
\[
  \max_{1\le m\le M}|\widehat f_n(u_m)-f_X(u_m)|\to 0
  \quad\hbox{almost surely}.
\]
Consequently,
\[
  D_{U,w}(\widehat f_n,f_X)\to 0
  \quad\hbox{almost surely}.
\]
\end{proposition}

\begin{proof}
For a fixed $u_m$, the random variables $\cos(u_mX_j)$ and $\sin(u_mX_j)$ are
bounded.  The strong law of large numbers gives almost sure convergence of
their sample means to their expectations.  Thus
$\widehat f_n(u_m)\to f_X(u_m)$ almost surely for each fixed $m$.  Since the
grid is finite, the maximum over $m=1,\ldots,M$ also converges to zero almost
surely.  The distance $D_{U,w}$ is a finite weighted sum of squared moduli, so
the final assertion follows.
\end{proof}

\begin{theorem}[Shared characteristic-function representation]\label{thm:compat}
Both regimes are finite estimating systems whose coefficients are
characteristic-function evaluations on a finite frequency set.  In the
finite-variance trigonometric regime, the Kunchenko normal system $FK=B$ has,
by Proposition~\ref{prop:trig}, entries determined by the values of $f_X$ and
$f'_X$ at integer multiples of the base frequency $p$.  In the moment-free regime, the empirical
vector $(\widehat f_n(u_1),\ldots,\widehat f_n(u_M))$ defines, by
Propositions~\ref{prop:bdd} and~\ref{prop:consist}, a bounded discrepancy from
reference CF vectors that is well defined and almost surely consistent without
moment assumptions.  Consequently both constructions reduce to finite systems
of equations in the same set of characteristic-function evaluations.
\end{theorem}

\begin{proof}
The finite-variance statement is Proposition~\ref{prop:trig}: the
orthogonality conditions defining $FK=B$ are linear moment conditions whose
coefficients are the CF and CF-derivative evaluations listed there.  The
moment-free statement combines Propositions~\ref{prop:bdd}
and~\ref{prop:consist}.  Both systems equate functionals of a finite set of
characteristic-function evaluations to target values, which is the asserted
common form.
\end{proof}

\begin{remark}
The representation is structural, not an equivalence.  The two regimes use the
same characteristic-function evaluations but different moment functions:
projection-residual moments for the covariance normal system, and CF-matching
moments for the discrepancy.  They therefore optimize different functionals,
and Theorem~\ref{thm:compat} unifies their algebraic form without reducing one
criterion to the other or asserting that either dominates.
\end{remark}

\section{Bounded CF scores}

The empirical CF geometry motivates bounded estimating equations.  For residuals
$r_j(\theta)$ and a frequency $u$, define
\[
  \Psi_n(\theta;u)=\frac{1}{n}\sum_{j=1}^n \sin(u r_j(\theta)).
\]
For a finite set of frequencies $u_1,\ldots,u_L$, a weak CF estimator may be
defined by solving $\Psi_n(\theta;u_\ell)=0$ or by minimizing
\[
  \widehat\theta
  =\argmin_\theta \sum_{\ell=1}^L a_\ell\Psi_n(\theta;u_\ell)^2,
  \qquad a_\ell>0.
\]
This paper uses this construction only as a bounded stochastic-polynomial
score, not as a complete estimator theory.

\begin{proposition}[Bounded odd score]\label{prop:score}
For all $u,r\in\R$,
\[
  |\sin(ur)|\le 1,\qquad \sin(u(-r))=-\sin(ur).
\]
Moreover, near the origin,
\[
  \sin(ur)=ur-\frac{(ur)^3}{3!}+\frac{(ur)^5}{5!}-\cdots,
\]
so the bounded score is linked to the odd polynomial hierarchy.
\end{proposition}

\begin{proof}
The bound and oddness are elementary trigonometric identities.  The series is
the standard Taylor expansion of sine.
\end{proof}

\begin{proposition}[Cauchy bridge]\label{prop:cauchy}
Let $\gamma\ne 0$ and
\[
  w_\gamma(r)=\frac{1}{1+(r/\gamma)^2}.
\]
Then
\[
  w_\gamma(r)r=\gamma^2\frac{r}{\gamma^2+r^2}.
\]
Thus the degree-one windowed score is proportional to the known-scale Cauchy
location score $2r/(\gamma^2+r^2)$.
\end{proposition}

\begin{proof}
The identity follows by multiplying numerator and denominator by $\gamma^2$:
\[
  \frac{r}{1+r^2/\gamma^2}
  =\frac{\gamma^2 r}{\gamma^2+r^2}.
\]
\end{proof}

\section{Minimum-CF-distance estimation and its asymptotics}
\label{sec:est}

The bounded discrepancy of Definition~\ref{def:cfdist} yields a parameter
estimator that needs no moments.  Let $\{f_\theta:\theta\in\Theta\}$, with
$\Theta\subset\R^d$, be a parametric family of characteristic functions.

\begin{definition}[Minimum-CF-distance estimator]\label{def:cfpmm}
For a fixed grid $U$ and positive weights $w$,
\[
  \widehat\theta_n
   =\argmin_{\theta\in\Theta} D_{U,w}(\widehat f_n,f_\theta)
   =\argmin_{\theta\in\Theta}\sum_{m=1}^M w_m
      |\widehat f_n(u_m)-f_\theta(u_m)|^2 .
\]
\end{definition}

This is the characteristic-function counterpart of the Kunchenko
stochastic-polynomial fit: the empirical CF vector
$(\widehat f_n(u_1),\ldots,\widehat f_n(u_M))$ plays the role of the sample
moments, and the weighted squared discrepancy replaces the covariance
criterion.  By Proposition~\ref{prop:bdd} the objective is bounded for every
law, with no moment assumption.  Write
\[
  Q_n(\theta)=D_{U,w}(\widehat f_n,f_\theta),\qquad
  Q(\theta)=\sum_{m=1}^M w_m|f_{\theta_0}(u_m)-f_\theta(u_m)|^2,
\]
where $\theta_0$ denotes the true parameter.

\begin{definition}[Grid identifiability]\label{def:ident}
The family $\{f_\theta\}$ is identifiable at $\theta_0$ on $U$ if
$f_\theta(u_m)=f_{\theta_0}(u_m)$ for all $m$ implies $\theta=\theta_0$;
equivalently, $Q(\theta)=0$ only at $\theta=\theta_0$.  (Injectivity of
$\theta\mapsto(f_\theta(u_1),\ldots,f_\theta(u_M))$ on $\Theta$ is a stronger
sufficient condition.)
\end{definition}

\begin{remark}
Grid identifiability is easy to verify for the standard heavy-tailed families.
For the Cauchy location--scale family $f_\theta(u)=\exp(\ii\delta u-\gamma|u|)$,
any two distinct frequencies $u_1\ne u_2$ with $u_m>0$ determine
$(\delta,\gamma)$, since $\log|f_\theta(u)|=-\gamma|u|$ fixes $\gamma$ and
$\arg f_\theta(u)=\delta u$ fixes $\delta$.  For the symmetric stable family
$f_\theta(u)=\exp(-(\gamma|u|)^\alpha)$, two distinct positive frequencies
determine $(\alpha,\gamma)$, because
$\log(-\log f_\theta(u))=\alpha\log\gamma+\alpha\log|u|$ is affine in
$\log|u|$ with slope $\alpha$.  A single frequency is not enough in general:
for the symmetric stable family the characteristic function is real, so one
frequency gives a single equation for $(\alpha,\gamma)$, while for the Cauchy
family one frequency fixes $\gamma$ but determines $\delta$ only modulo
$2\pi/u$, an ambiguity that a second frequency removes.
Identifiability is thus a joint property of the family and the grid.  The next
proposition records a general, checkable sufficient condition; it is the key
assumption of the consistency result that follows.
\end{remark}

\begin{proposition}[Grid identifiability from smoothness and rank]\label{prop:ident-gen}
Let
\[
  \Phi_U(\theta)=\big(\Rea f_\theta(u_1),\Ima f_\theta(u_1),\ldots,
   \Rea f_\theta(u_M),\Ima f_\theta(u_M)\big)
\]
map $\Theta\subset\R^d$ into $\R^{2M}$.
\textup{(i)} If $\Phi_U$ is continuously differentiable on a neighbourhood of
$\theta_0$ and its $2M\times d$ Jacobian $G(\theta_0)$ has full column rank $d$,
then $\theta_0$ is locally identified on $U$: $\Phi_U(\theta)=\Phi_U(\theta_0)$
for $\theta$ near $\theta_0$ forces $\theta=\theta_0$.
\textup{(ii)} If, in addition, $\Phi_U$ is injective on the compact set
$\Theta$, then the family is grid identifiable at $\theta_0$ on $U$ in the sense
of Definition~\ref{def:ident}.  A $d$-parameter family is therefore generically
identified by $\lceil d/2\rceil$ distinct nonzero frequencies, each contributing
the two coordinates $\Rea f_\theta(u_m)$ and $\Ima f_\theta(u_m)$.
\end{proposition}

\begin{proof}
\textup{(i)} Full column rank of the continuous Jacobian $G(\theta_0)$ makes
$\Phi_U$ an immersion at $\theta_0$, hence locally injective by the rank theorem.
\textup{(ii)} Global injectivity of $\Phi_U$ on $\Theta$ is exactly the stronger
sufficient condition of Definition~\ref{def:ident}, so $Q(\theta)=0$ only at
$\theta_0$.  An immersion requires $2M\ge d$, i.e. $M\ge\lceil d/2\rceil$
\cite{rothenberg1971}.
\end{proof}

The rank condition in part \textup{(i)} is the same full-column-rank requirement
imposed on $G$ in Proposition~\ref{prop:an-est}, so local identifiability is
available wherever the asymptotic-normality result applies; the Cauchy and
stable computations of the preceding remark are the instances $d=2$,
$\lceil d/2\rceil=1$, where one frequency leaves a residual ambiguity that a
second distinct frequency removes.

\begin{proposition}[Consistency on a fixed grid]\label{prop:cons-est}
Let $X_1,X_2,\ldots$ be i.i.d.\ with characteristic function $f_{\theta_0}$,
and assume \textup{(i)} $\Theta$ is compact; \textup{(ii)} for each $m$ the map
$\theta\mapsto f_\theta(u_m)$ is continuous on $\Theta$; \textup{(iii)} the
family is identifiable at $\theta_0$ on $U$.  Then every minimizer $\widehat\theta_n$ of
Definition~\ref{def:cfpmm} satisfies $\widehat\theta_n\to\theta_0$ almost
surely.  No moment assumption is used.
\end{proposition}

\begin{proof}
For complex numbers $a,b,c$ of modulus at most one,
\[
  \big||a-c|^2-|b-c|^2\big|
   =\big||a-c|-|b-c|\big|\,\big(|a-c|+|b-c|\big)\le 4|a-b|,
\]
by the reverse triangle inequality and $|a-c|+|b-c|\le4$.  Applying this with
$a=\widehat f_n(u_m)$, $b=f_{\theta_0}(u_m)$ and $c=f_\theta(u_m)$ gives,
uniformly in $\theta$,
\[
  \sup_{\theta\in\Theta}|Q_n(\theta)-Q(\theta)|
   \le 4\sum_{m=1}^M w_m\,|\widehat f_n(u_m)-f_{\theta_0}(u_m)|,
\]
which tends to zero almost surely by Proposition~\ref{prop:consist}.  The limit
$Q$ is continuous and nonnegative, and by identifiability $Q(\theta)>0$ for
$\theta\ne\theta_0$ while $Q(\theta_0)=0$.  Since $Q$ is continuous and the set
$\{\theta\in\Theta:\|\theta-\theta_0\|\ge\varepsilon\}$ is compact,
$\inf_{\|\theta-\theta_0\|\ge\varepsilon}Q(\theta)>0$ for every $\varepsilon>0$,
so the minimizer is well separated.  Uniform almost sure convergence of $Q_n$
to $Q$ over the compact set $\Theta$ then gives $\widehat\theta_n\to\theta_0$
almost surely \cite[Theorem~5.7]{vandervaart2000}.
\end{proof}

The fixed-grid limit law below uses one regularity hypothesis on the
\emph{model}, stated once and invoked by the asymptotic-normality and efficiency
results.  It governs the parameter dependence $\theta\mapsto P_\theta$ and, as
the following lemma records, lets bounded characteristic-function moments be
differentiated under the expectation; it is not a condition on the tails of $X$.

\begin{assumption}[Regularity: differentiability in quadratic mean]\label{asm:dqm}
The model $\{P_\theta:\theta\in\Theta\}$ has densities $p_\theta$ with respect to
a $\sigma$-finite measure $\mu$ and is differentiable in quadratic mean at the
true parameter $\theta_0$ \cite[Def.~7.1]{vandervaart2000}: there is a score
$s=s_{\theta_0}\in L^2(P_{\theta_0})^d$ with
\[
  \int\Big(\sqrt{p_{\theta_0+h}}-\sqrt{p_{\theta_0}}
    -\tfrac12 h^\top s\,\sqrt{p_{\theta_0}}\Big)^2 d\mu=o(\|h\|^2),
  \qquad h\to0,
\]
and the Fisher information $I(\theta_0)=\E_{\theta_0}[s\,s^\top]$ is finite and
nonsingular.
\end{assumption}

Assumption~\ref{asm:dqm} restricts only the smoothness of
$\theta\mapsto p_\theta$ in the parameter and imposes no moment condition on $X$;
in particular $I(\theta_0)$ is finite for the location--scale Cauchy and the
symmetric stable families even though $X$ has no variance.  It is verifiable
through the standard sufficient condition \cite[Lemma~7.6]{vandervaart2000}: it
holds whenever $\theta\mapsto\sqrt{p_\theta(x)}$ is continuously differentiable
for $\mu$-almost every $x$ and the map $\theta\mapsto I(\theta)$ is well defined,
finite, and continuous, both direct to check for the smooth heavy-tailed
families considered here.

\begin{lemma}[Differentiation of bounded test functions]\label{lem:btf}
Under Assumption~\ref{asm:dqm}, for every bounded measurable $\psi:\R\to\R$ the
map $\theta\mapsto\E_\theta\psi(X)$ is differentiable at $\theta_0$ with
\[
  \partial_\theta\,\E_\theta\psi(X)\big|_{\theta_0}=\E_{\theta_0}[\psi(X)\,s].
\]
In particular $\E_{\theta_0}s=0$ \textup{(}take $\psi\equiv1$\textup{)}, so for
the bounded functions $\cos(u\cdot),\sin(u\cdot)$ the interchange of derivative
and expectation needed below is justified.
\end{lemma}

\begin{proof}
Differentiability in quadratic mean is the $L^2(\mu)$ statement
$\sqrt{p_{\theta_0+h}}=\sqrt{p_{\theta_0}}\,(1+\tfrac12 h^\top s)+r_h$ with
$\|r_h\|_{L^2(\mu)}=o(\|h\|)$.  For bounded $\psi$ expand
$\E_{\theta_0+h}\psi=\int\psi\,(\sqrt{p_{\theta_0+h}})^2 d\mu$: the constant term
is $\E_{\theta_0}\psi$, the linear term is
$h^\top\!\int\psi\,s\,p_{\theta_0}\,d\mu=h^\top\E_{\theta_0}[\psi s]$, and every
remaining term is bounded by $\|\psi\|_\infty$ times a multiple of
$\|r_h\|_{L^2(\mu)}$ or of $\|h\|^2$, hence $o(\|h\|)$, by Cauchy--Schwarz and
$\|\sqrt{p_{\theta_0}}\|_{L^2(\mu)}=1$.  Thus
$\E_{\theta_0+h}\psi=\E_{\theta_0}\psi+h^\top\E_{\theta_0}[\psi s]+o(\|h\|)$.
Taking $\psi\equiv1$ gives $\E_{\theta_0}s=0$ from
$\partial_\theta\!\int p_\theta\,d\mu=0$.
\end{proof}

\begin{proposition}[Asymptotic normality and optimal weighting]\label{prop:an-est}
Assume in addition that $\theta_0$ is interior to $\Theta$, that
$\theta\mapsto f_\theta(u_m)$ is continuously differentiable near $\theta_0$,
and that every $u_m\ne0$ and $X$ is nondegenerate.  Let $\widehat m_n$ and
$m(\theta)$ stack the parts
$(\Rea\widehat f_n(u_m),\Ima\widehat f_n(u_m))_{m=1}^M$ and
$(\Rea f_\theta(u_m),\Ima f_\theta(u_m))_{m=1}^M$, let
$G=\partial m/\partial\theta^\top$ at $\theta_0$ be the $2M\times d$ Jacobian of
full column rank $d$, and let $\Omega$ be the asymptotic covariance of
$\sqrt n\,\widehat m_n$, assumed nonsingular
(Proposition~\ref{prop:omega} below gives a verifiable sufficient condition).
Under Assumption~\ref{asm:dqm} the differentiability assumed here holds and the
Jacobian is $G=\Cov(m,s)$ by Lemma~\ref{lem:btf}, while $\Omega=\Cov(m)$ is the
covariance of the bounded moment vector.  For any symmetric
positive-definite weight matrix $W$ (the present estimator uses the diagonal
$W=\operatorname{diag}(w_m,w_m)$),
\[
  \sqrt n\,(\widehat\theta_n-\theta_0)\xrightarrow{d}
   \mathcal N\!\Big(0,\,(G^\top WG)^{-1}G^\top W\Omega WG\,(G^\top WG)^{-1}\Big).
\]
The entries of $\Omega$ are covariances of the bounded variables $\cos(u_mX)$
and $\sin(u_mX)$, hence finite for every law, in particular for Cauchy and
symmetric stable laws.  The weight $W=\Omega^{-1}$ minimizes the asymptotic
covariance, which then equals $(G^\top\Omega^{-1}G)^{-1}$.
\end{proposition}

\begin{proof}
The estimator minimizes
$(\widehat m_n-m(\theta))^\top W(\widehat m_n-m(\theta))$, so it is the
minimum-distance (generalized method-of-moments) estimator \cite{hansen1982}
for the CF-matching conditions $\E\cos(u_mX)=\Rea f_\theta(u_m)$,
$\E\sin(u_mX)=\Ima f_\theta(u_m)$.  Its first-order condition
$G^\top W(\widehat m_n-m(\widehat\theta_n))=0$, the consistency of
$\widehat\theta_n$ (Proposition~\ref{prop:cons-est}), and the expansion
$m(\widehat\theta_n)=m(\theta_0)+G(\widehat\theta_n-\theta_0)
 +o_p(\|\widehat\theta_n-\theta_0\|)$ yield
$\sqrt n(\widehat\theta_n-\theta_0)
 =(G^\top WG)^{-1}G^\top W\sqrt n(\widehat m_n-m(\theta_0))+o_p(1)$.
Since $\widehat m_n$ is a sample mean of the bounded vector
$(\cos(u_mX),\sin(u_mX))_m$, the multivariate central limit theorem gives
$\sqrt n(\widehat m_n-m(\theta_0))\xrightarrow{d}\mathcal N(0,\Omega)$ with
$\Omega$ finite and free of moment assumptions; the sandwich form, and the
optimal weight $W=\Omega^{-1}$, follow \cite[Ch.~5]{vandervaart2000}.
\end{proof}

The nonsingularity of $\Omega$ used above is not an extra hypothesis but a
verifiable property of the grid, and it has a transparent linear-algebraic
characterisation.

\begin{proposition}[Nonsingularity of the moment covariance]\label{prop:omega}
Let $X$ be nondegenerate with a distribution whose support has an accumulation
point \textup{(}in particular whenever $X$ has a density, or any infinite
support\textup{)}, and let $u_1,\ldots,u_M$ be distinct and nonzero.  Then the
$2M$ centred functions
$\{\cos(u_m\cdot)-\E\cos(u_mX),\ \sin(u_m\cdot)-\E\sin(u_mX)\}_{m=1}^M$ are
linearly independent in $L^2(P_X)$, so the moment covariance $\Omega(U)$ is
nonsingular.  Equivalently, $\Omega(U)$ is singular if and only if some
nontrivial real combination $\sum_{m}\{a_m\cos(u_mx)+b_m\sin(u_mx)\}$ is
$P_X$-almost surely equal to a constant.
\end{proposition}

\begin{proof}
$\Omega(U)$ is the Gram matrix of the centred system in $L^2(P_X)$, hence
singular precisely when a nontrivial real combination
$\sum_m\{a_m\cos(u_mx)+b_m\sin(u_mx)\}-c$ vanishes $P_X$-almost surely.  Such a
finite trigonometric sum is a real-analytic \textup{(}indeed entire\textup{)}
function of $x$; if it vanishes on a set with an accumulation point it vanishes
identically on $\R$.  Distinct nonzero frequencies make
$\{1,\cos(u_m\cdot),\sin(u_m\cdot)\}_m$ linearly independent as functions on
$\R$, so all $a_m,b_m,c$ vanish.  Hence no nontrivial combination is almost
surely constant and $\Omega(U)\succ0$.  The argument uses only nondegeneracy of
$X$; no moment of $X$ enters.
\end{proof}

\begin{remark}[Conditioning and the excluded frequency]\label{rem:omega}
The off-diagonal entries of $\Omega(U)$ are real and imaginary parts of
$f_X(u_m\pm u_k)$, so two frequencies with $|u_m-u_k|\cdot s\ll1$, where $s$ is a
scale of $X$, give nearly parallel columns and an ill-conditioned $\Omega(U)$
even though it is invertible in the limit.  A practical rule is to keep the
spacing $\Delta u\gtrsim1/s$, with $s$ a robust dispersion such as the median
absolute deviation, and to monitor the condition number or smallest eigenvalue
of the sample $\widehat\Omega$, shrinking it toward its diagonal when this is
large \textup{(}Remark~\ref{rem:grid}\textup{)}.  The frequency $u=0$ is excluded
for the same reason: $f_X(0)=1$ for every law, so $\cos(0\cdot)\equiv1$ is the
constant already removed by centring, and including it makes the centred system
contain the zero vector and $\Omega(U)$ singular.
\end{remark}

The fixed-grid estimator is $\sqrt n$-consistent but not in general efficient.
The next two results show that efficiency is recovered as the grid is refined.
Write $J(U)=G^\top\Omega^{-1}G$ for the optimal-weight (efficient) information of
Proposition~\ref{prop:an-est}, so that $J(U)^{-1}$ is the smallest asymptotic
covariance attainable on the grid $U$.  Throughout we take the frequencies
nonzero, so that $\Omega(U)$ is nonsingular by Proposition~\ref{prop:omega}.
Recall the score $s=s_{\theta_0}$ of Assumption~\ref{asm:dqm}, and
for each $u$ let $m_u^{\mathrm c}=\cos(uX)-\E\cos(uX)$ and
$m_u^{\mathrm s}=\sin(uX)-\E\sin(uX)$ be the centred moment functions, with
$H(U)$ their closed linear span in $L^2(P_{\theta_0})$.

\begin{lemma}[Monotone efficiency]\label{lem:monotone}
For every finite grid $U$ of nonzero frequencies and every $a\in\R^d$,
$a^\top J(U)\,a=\|\Pi_{H(U)}(a^\top s)\|^2_{L^2(P_{\theta_0})}$, the squared norm
of the orthogonal projection of the scalar score combination $a^\top s$ onto
$H(U)$.  Consequently $J(U)\preceq J(U')$ in the Loewner order whenever
$U\subseteq U'$: refining or extending the grid never decreases the efficient
information.
\end{lemma}

\begin{proof}
Because $\cos(uX),\sin(uX)$ are bounded, Lemma~\ref{lem:btf} (under
Assumption~\ref{asm:dqm}) gives
$\partial_\theta\E_\theta\cos(uX)=\E[\cos(uX)\,s]$ and
$\partial_\theta\E_\theta\sin(uX)=\E[\sin(uX)\,s]$, the interchange of derivative
and expectation for these characteristic-function moments.
Hence the Jacobian is $G=\operatorname{Cov}(m,s)$ (using $\E s=0$) and the moment
covariance is $\Omega=\operatorname{Cov}(m)$, where $m$ stacks the centred
$m_u^{\mathrm c},m_u^{\mathrm s}$ over $u\in U$ (centring shifts the
$\widehat m_n,m(\theta)$ of Proposition~\ref{prop:an-est} by a deterministic
constant, so $G$ and $\Omega$ are unchanged).  Thus
$J(U)=G^\top\Omega^{-1}G
=\operatorname{Cov}(s,m)\operatorname{Cov}(m)^{-1}\operatorname{Cov}(m,s)$, and
for any $a\in\R^d$,
\[
  a^\top J(U)\,a
   =\operatorname{Cov}(a^\top s,m)\operatorname{Cov}(m)^{-1}
     \operatorname{Cov}(m,a^\top s)
   =\|\Pi_{H(U)}(a^\top s)\|^2,
\]
the last equality being the least-squares identity for the projection of
$a^\top s$ onto $\operatorname{span}(m)=H(U)$.  For $U\subseteq U'$ we have
$H(U)\subseteq H(U')$, and the projection norm is monotone under enlargement of
the subspace (Pythagoras; this step is independently checked in Lean,
Table~\ref{tab:lean}), so $a^\top J(U)\,a\le a^\top J(U')\,a$ for all $a$,
i.e. $J(U)\preceq J(U')$.
\end{proof}

\begin{theorem}[Efficiency in the dense-grid limit]\label{thm:dense}
Suppose Assumption~\ref{asm:dqm} holds, with score $s\in L^2(P_{\theta_0})$ and
nonsingular Fisher information $I(\theta_0)$.
Let $U_1\subseteq U_2\subseteq\cdots$ be finite grids of nonzero frequencies
whose union is dense in $\R$.  Then $J(U_M)\uparrow I(\theta_0)$ in the Loewner
order, and hence the optimally weighted minimum-CF-distance estimator has
asymptotic covariance $J(U_M)^{-1}\downarrow I(\theta_0)^{-1}$: it is
asymptotically efficient in the limit of a dense grid.
\end{theorem}

\begin{proof}
Fix $a\in\R^d$.  By Lemma~\ref{lem:monotone},
$a^\top J(U_M)\,a=\|\Pi_{H(U_M)}(a^\top s)\|^2$ is nondecreasing in $M$ and
bounded above by $\|a^\top s\|^2=a^\top I(\theta_0)a$.  The maps
$u\mapsto\cos(uX),\sin(uX)$ are continuous into $L^2(P_{\theta_0})$, since
$\E|e^{\ii uX}-e^{\ii vX}|^2=2(1-\Rea f_{\theta_0}(u-v))\to0$ as $v\to u$; hence
the closed span of the centred system over $\bigcup_M U_M$ equals its closed
span over all nonzero frequencies.  The real system
$\{\cos(uX),\sin(uX):u\in\R\}$ is total in $L^2(P_{\theta_0})$: if
$g\in L^2(P_{\theta_0})$ is orthogonal to all of them, then
$g\in L^1(P_{\theta_0})$ by Cauchy--Schwarz, and
$\E[g\,e^{\ii uX}]=\E[g\cos(uX)]+\ii\,\E[g\sin(uX)]=0$ for every $u$, so the
finite signed measure $g\,dP_{\theta_0}$ has identically vanishing Fourier
transform and is therefore null (uniqueness of the Fourier transform of a finite
measure), whence $g=0$.  Because the constant $1$ is orthogonal to every centred
$m_u^{\mathrm c},m_u^{\mathrm s}$ and $\cos(uX)=m_u^{\mathrm c}+\E\cos(uX)$, the
orthogonal decomposition
$L^2(P_{\theta_0})=\R\!\cdot\!1\oplus\overline{\operatorname{span}}\,
\{m_u^{\mathrm c},m_u^{\mathrm s}:u\in\R\}$ identifies the closed centred span
with $L^2_0(P_{\theta_0})$; with the previous sentence, $\bigcup_M H(U_M)$ is
dense in $L^2_0(P_{\theta_0})$.  Each component of $s$ lies in
$L^2_0(P_{\theta_0})$ (it has mean zero), so projection onto the increasing
subspaces converges: $\Pi_{H(U_M)}(a^\top s)\to a^\top s$, whence
$a^\top J(U_M)\,a\uparrow a^\top I(\theta_0)a$.  As $a$ was arbitrary,
$J(U_M)\uparrow I(\theta_0)$ in the Loewner order, and inversion reverses it.
\end{proof}

\begin{remark}
Theorem~\ref{thm:dense} recovers, by an elementary projection argument, the
asymptotic efficiency of empirical-characteristic-function estimation
established for a continuum of frequencies by Feuerverger and
McDunnough~\cite{feuerverger1981}; the operator-theoretic optimal weighting of
the continuum limit is treated by Carrasco and Florens~\cite{carrasco2000}.
\end{remark}

\begin{remark}[Choice of grid and weights]\label{rem:grid}
A bounded grid with $u_{\max}\approx\pi/s$, where $s$ is a robust scale such as
the median absolute deviation, and with $M$ between roughly $20$ and $64$
log-spaced points is adequate; large frequencies carry little information,
because $|f_\theta(u)|$ is then small and the empirical CF is dominated by
sampling noise.  If the grid scales with a robust scale estimate, the estimator
is location--scale equivariant.  Optimal weighting requires an estimate of
$\Omega$, which is noisy for small $n$; uniform weights are then preferable, a
standard small-sample caveat for the generalized method of moments
\cite{hansen1982}.
\end{remark}

\section{Lean 4 verification scope}

Lean 4~\cite{demoura2021} and its mathematical library
Mathlib~\cite{mathlib2020} are used only to audit deterministic identities.  The
self-contained \texttt{KunchenkoCF} library accompanying this paper is provided
in the supplement (\texttt{lean\_project/}); together with the replication
scripts of the numerical illustration it is publicly available at
\url{https://github.com/SZabolotnii/Ku_CF-code-supplement}.  The command
\texttt{lake build KunchenkoCF} completes successfully with Lean \texttt{v4.26.0}
and Mathlib
\texttt{v4.26.0}, with no \texttt{sorry} and no added \texttt{axiom}.  Besides
the bounded-score identities it includes a machine-checked proof that the
orthogonal-projection norm is monotone under subspace inclusion, the
deterministic core of Lemma~\ref{lem:monotone}.

\begin{table}[h]
\caption{Lean lemmas used as deterministic audit points.}
\label{tab:lean}
\centering
\scriptsize
\begin{tabular}{p{0.31\linewidth}p{0.49\linewidth}p{0.12\linewidth}}
\toprule
Lean lemma & Manuscript role & Scope\\
\midrule
\texttt{abs\_cfScore\_le\_one} & boundedness of $\sin(ur)$ & algebraic\\
\texttt{cfScore\_odd} & oddness of the residual score & algebraic\\
\texttt{cauchyWeight\_mul\_eq\_mle} & Cauchy bridge identity & algebraic\\
\texttt{ident\_of\_unit\_det} & linear identifiability core & linear\\
\texttt{norm\_starProjection\_mono} & monotone projection norm (Lemma~\ref{lem:monotone}) & linear\\
\bottomrule
\end{tabular}
\end{table}

The Lean supplement does not formalize characteristic-function differentiability,
complex ECF distance, the strong law of large numbers, asymptotic normality, or
global statistical identifiability.  These statements are deliberately kept in
the non-formalized proof layer.  In particular, the finite-grid strong law
underlying Proposition~\ref{prop:consist} is a direct corollary of the Mathlib
strong law for bounded random variables and could be formalized; we keep it in
the non-formalized layer so that the formal audit remains purely algebraic.

\section{Numerical illustration}

The numerical example is included only to show the behavior of the finite-grid
definition.  We used the grid
$U=\{0.2,0.4,\ldots,2.4\}$ with uniform weights and computed
$D_{U,w}(\widehat f_n,f)$ for three distributions with known characteristic
functions:
\[
  f_G(u)=\exp(-u^2/2),\qquad
  f_C(u)=\exp(-|u|),\qquad
  f_\alpha(u)=\exp(-|u|^{1.3}).
\]
For each sample size, 300 independent replications were used.  Figure
\ref{fig:ecf} shows that the finite-grid discrepancy decreases with $n$ for
Gaussian, Cauchy, and symmetric stable samples.  This is an illustration of
Proposition~\ref{prop:consist}, not a benchmark of a full estimator.

\begin{figure}[h]
\centering
\includegraphics[width=0.88\linewidth]{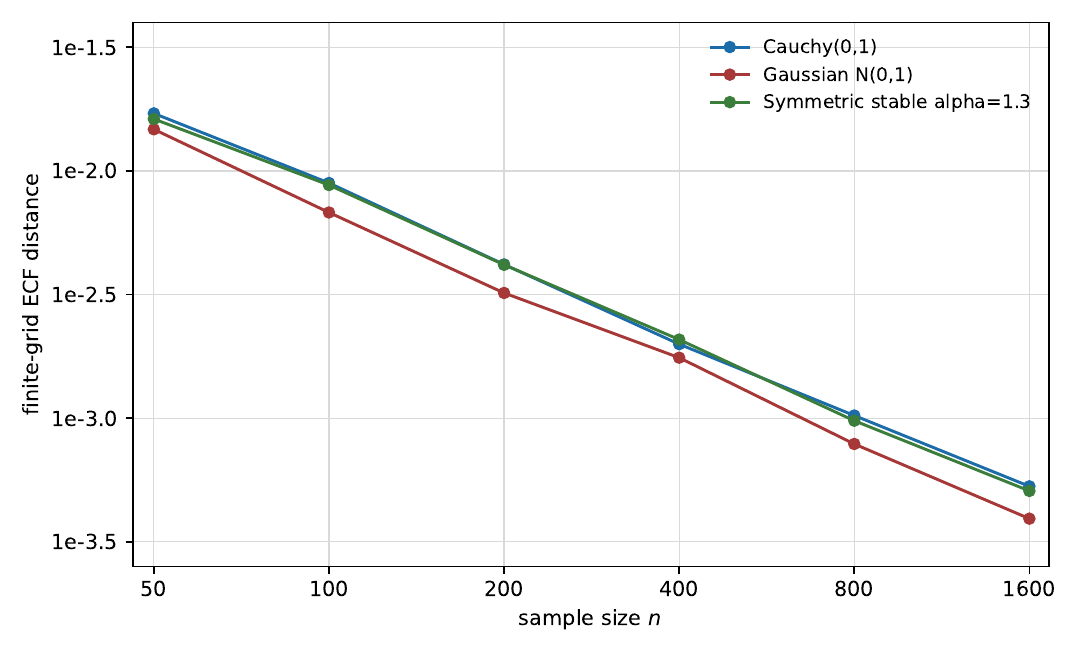}
\caption{Finite-grid ECF distance for Gaussian, Cauchy, and symmetric
stable ($\alpha=1.3$) samples.  The plot uses the closed-form characteristic
function of each law and a fixed 12-point frequency grid.}
\label{fig:ecf}
\end{figure}

\begin{table}[h]
\caption{Moment status of representative laws.}
\label{tab:moments}
\centering
\begin{tabular}{p{0.28\linewidth}ccc p{0.25\linewidth}}
\toprule
Law & Mean & Variance & Fourth moment & Safe language\\
\midrule
Gaussian $N(0,1)$ & yes & yes & yes & covariance and MSE\\
Student $t_5$ & yes & yes & yes & finite fourth moment\\
Student $t_3$ & yes & yes & no & no kurtosis route\\
Cauchy $(0,1)$ & no & no & no & CF distance only\\
Symmetric stable $\alpha=1.3$ & yes & no & no & mean only; CF distance\\
\bottomrule
\end{tabular}
\end{table}

We also illustrate the minimum-CF-distance estimator of
Definition~\ref{def:cfpmm}.  For the Cauchy location--scale family and the
symmetric stable family we drew $300$ independent samples at each size
$n\in\{100,200,400,800,1600\}$, minimized $D_{U,w}(\widehat f_n,f_\theta)$ over
a median-absolute-deviation grid of $24$ frequencies with uniform weights, and
compared the result with standard references: the median and the interquartile
range for the Cauchy location and scale, the
Koutrouvelis~\cite{koutrouvelis1980} regression for the stable parameters, and
the sample mean and standard deviation.  Table~\ref{tab:est} reports the
root-mean-square error.  The minimum-CF-distance estimator is essentially
unbiased and its error decreases at the parametric rate, in agreement with
Propositions~\ref{prop:cons-est} and~\ref{prop:an-est}; it is competitive with
the purpose-built robust and regression references, while the sample-moment
estimators diverge, having no finite target for these laws.
Figure~\ref{fig:est} shows the stable-index error against $n$ on a log--log
scale for the CF estimator and the Koutrouvelis regression: both decrease at the
same rate, with the regression slightly more accurate at these sample sizes,
consistent with the fixed-grid, uniform-weight choice discussed in
Remark~\ref{rem:grid}.  Estimates from linear, logarithmic, and
median-absolute-deviation grids agree in root-mean-square error to within six
percent, so the conclusions are not an artifact of the grid.

\begin{figure}[h]
\centering
\includegraphics[width=0.88\linewidth]{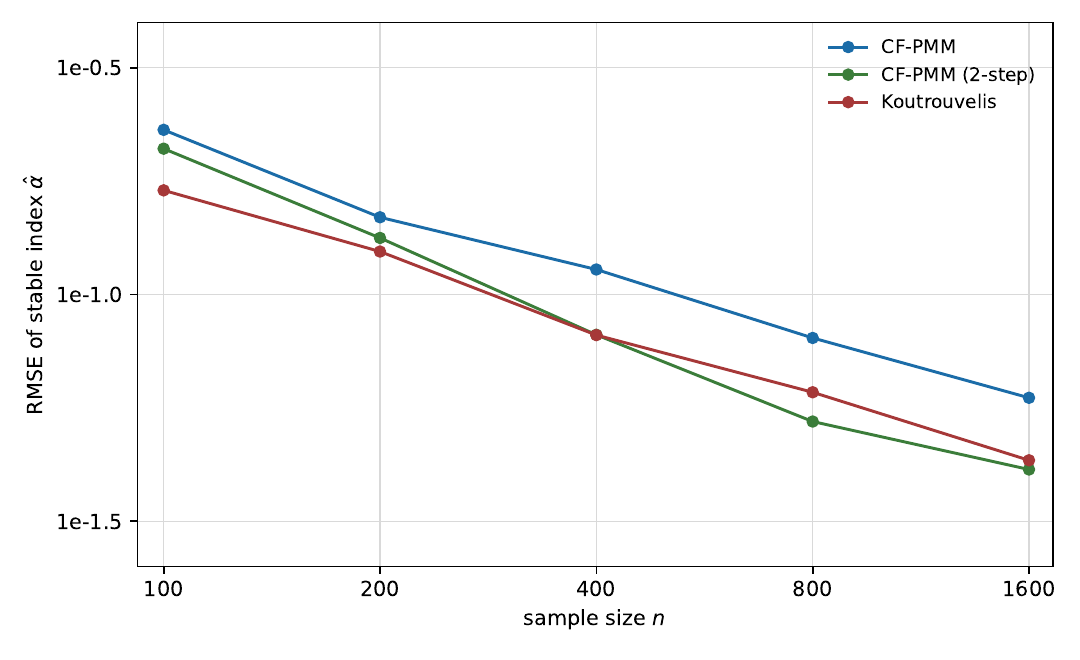}
\caption{Root-mean-square error of the symmetric stable index estimate
$\widehat\alpha$ (true $\alpha=1.3$) against sample size $n$, on a log--log
scale, for the uniform-weight minimum-CF-distance estimator (CF-PMM), its
two-step optimally-weighted version, and the Koutrouvelis regression
($300$ replications).}
\label{fig:est}
\end{figure}

\begin{table}[h]
\caption{Monte Carlo root-mean-square error of parameter estimates
($300$ replications).  The sample-moment estimators are included only to show
the breakdown of moment-based methods; they have no finite target for these
laws.}
\label{tab:est}
\centering
\small
\begin{tabular}{llcc}
\toprule
Law and parameter & Estimator & $n=200$ & $n=800$\\
\midrule
Cauchy location $\delta$ & CF-PMM & 0.113 & 0.054\\
                         & median & 0.117 & 0.055\\
                         & sample mean & 6.81 & 82.5\\
Cauchy scale $\gamma$    & CF-PMM & 0.105 & 0.056\\
                         & IQR$/2$ & 0.105 & 0.058\\
\midrule
Stable index $\alpha$    & CF-PMM & 0.148 & 0.080\\
                         & CF-PMM (2-step) & 0.133 & 0.053\\
                         & Koutrouvelis & 0.124 & 0.061\\
Stable scale $\gamma$    & CF-PMM & 0.077 & 0.041\\
                         & Koutrouvelis & 0.089 & 0.044\\
\bottomrule
\end{tabular}
\end{table}

To probe the asymptotics of Proposition~\ref{prop:an-est} we add two checks.
First, the two-step estimator with the optimal weight $W=\Omega^{-1}$ (the
sample moment covariance shrunk toward its diagonal for numerical stability)
lowers the stable-index error below both the uniform-weight estimator and the
Koutrouvelis regression once $n\ge800$ (Figure~\ref{fig:est} and
Table~\ref{tab:est}: RMSE $0.053$ against $0.080$ and $0.061$ at $n=800$), as
the efficiency statement predicts.  Second, percentile-bootstrap $95\%$
confidence intervals ($120$ resamples) have empirical coverage between $0.89$
and $0.96$ across both families and all parameters at $n=400$ and $n=1600$,
close to the nominal level and consistent with the asymptotic normality of
Proposition~\ref{prop:an-est}.  At small $n$ the optimally-weighted estimator
needs the diagonal shrinkage to avoid the ill-conditioning of $\Omega$, in line
with Remark~\ref{rem:grid}.

A final ablation \textup{(}in the replication script\textup{)} probes the grid
and the weighting.  Raising the upper frequency from $0.5\,\pi/\mathrm{MAD}$ to
$2\,\pi/\mathrm{MAD}$ keeps the stable-index error in the narrow band
$0.087$--$0.123$ at $n=400$, the smaller grids slightly better: large
frequencies never destabilise the estimator, since the empirical characteristic
function has variance at most $(1-|f_\theta(u)|^2)/n\le1/n$ at every $u$, so
high-frequency noise is bounded and merely uninformative
\textup{(}Remark~\ref{rem:grid}\textup{)}.  On the dense grid $\Omega$ is
numerically singular \textup{(}condition number $\approx10^{17}$\textup{)}, the
structural near-collinearity of Proposition~\ref{prop:omega} and
Remark~\ref{rem:omega}; diagonal shrinkage of intensity $0.6$ lowers the
condition number to about $55$ and the stable-index error to $0.084$, whereas the
unshrunk inverse gives the far larger error $0.70$.  A data-driven Ledoit--Wolf
intensity is small \textup{(}$\approx0.02$\textup{)} because the cosine and sine
correlations are genuine structure rather than sampling noise, so it
under-shrinks; the asymptotics make any consistent weight first-order optimal
\textup{(}Proposition~\ref{prop:an-est}\textup{)}, but at finite $n$ a fixed
diagonal or ridge shrinkage is the effective stabiliser of the ill-conditioned
grid, which is why uniform weights remain the practical default until $n$ is
large \textup{(}Remark~\ref{rem:grid}\textup{)}.

\section{Discussion and limitations}

The construction connects the finite-variance trigonometric Kunchenko normal
system with a moment-free empirical CF geometry.  It also gives a bounded-score path toward
weak stochastic-polynomial estimating equations.  The bridge is useful because
it prevents a common overclaim: a bounded sine basis and a CF-distance geometry
are related but not identical.  A simple trigonometric basis is not by itself a
CF-aware method unless the construction uses empirical or model characteristic
functions explicitly.

Several extensions are natural but outside this paper.  We have treated the
minimum-CF-distance estimator on a fixed grid and proved its efficiency in the
dense-grid limit (Theorem~\ref{thm:dense}); a data-driven optimal weighting and
the operator-theoretic continuum analogue remain to be developed.  Dependent
data are a further natural extension: because the moment functions
$\cos(uX),\sin(uX)$ are bounded, the estimating equations obey a central limit
theorem under standard strong-mixing conditions, with $\Omega$ replaced by a
long-run \textup{(}heteroskedasticity- and autocorrelation-consistent\textup{)}
covariance and consistency inherited from an ergodic strong
law~\cite{bradley2005}; the dense-grid efficiency statement, however, would
require the dependent-data information bound in place of the independent-case
Fisher information and is left open.  Heavy-tailed autoregressive models with
Cauchy innovations, a representative dependent-data instance, have been treated
by tailored expectation--maximization schemes~\cite{dhull2022}.
A characteristic-function classification theory would use classwise
log-CF-distance features, and a characteristic-function change-point theory
would develop sequential contrasts and their false-alarm calibration.  Neither
is claimed here.  The present
paper is confined to the characteristic-function geometry and the
minimum-CF-distance estimator for parametric families, and its verification
supplement is self-contained.

The Lean component should also be read conservatively.  It checks selected
deterministic identities used in the bounded-score construction; it does not
replace probabilistic proof.

\section{Conclusion}

Characteristic functions give a mathematically natural way to extend
Kunchenko stochastic-polynomial methods from moment-based finite-variance
settings to moment-free heavy-tailed settings.  In the finite-variance
trigonometric case, the covariance normal system can be expressed
through $f_X$ and $f'_X$.  In the moment-free case, empirical characteristic
functions on finite grids define bounded discrepancies without requiring raw
moments.  This yields a cautious but useful core: characteristic-function
geometry first, Kunchenko stochastic-polynomial interpretation second, Lean as
a deterministic audit trail third, and simulations only as illustration.

\end{document}